\documentclass[journal,twoside,web]{ieeecolor}

\usepackage{generic}
\usepackage{cite}
\usepackage{amsmath,amssymb,amsfonts}
\usepackage{algorithmic}
\usepackage{graphicx}
\usepackage{textcomp}

\usepackage{comment}

\usepackage{lcsys}

\usepackage[caption=false,font=footnotesize]{subfig}

\newtheorem{theorem}{Theorem}
\newtheorem{lemma}{Lemma}

\newtheorem{remark}{Remark}

\newtheorem{assumption}{Assumption}

\begin{document}

\def\BibTeX{{\rm B\kern-.05em{\sc i\kern-.025em b}\kern-.08em
    T\kern-.1667em\lower.7ex\hbox{E}\kern-.125emX}}
\markboth{\journalname, VOL. XX, NO. XX, XXXX 2017}
{Author \MakeLowercase{\textit{et al.}}: Preparation of Papers for IEEE Control Systems Letters (August 2022)}

\title{Dynamic Average Consensus with Privacy Guarantees and Its Application to Battery Energy Storage Systems}

\author{
Mihitha Maithripala, Chenyang Qiu, and
Zongli Lin
\thanks{Work supported in part by the US National Science Foundation under the grant number CMMI–2243930.}
\thanks{Mihitha Maithripala, Chenyang Qiu, and Zongli Lin are with the Charles L. Brown Department of Electrical and
Computer Engineering, University of Virginia, Charlottesville, VA 22904, U.S.A (e-mail:
wpg8hm@virginia.edu; zl5y@virginia.edu. Corresponding author:
Zongli Lin.)}
}

\maketitle
\thispagestyle{empty}
\pagestyle{plain}

\begin{abstract}
A privacy-preserving dynamic average consensus (DAC) algorithm is proposed that achieves consensus while preventing external eavesdroppers from inferring the reference signals and their derivatives. During the initialization phase, each agent generates a set of sinusoidal signals with randomly selected frequencies and exchanges them with its neighboring agents to construct a masking signal. Each agent masks its reference signals using this composite masking signal before executing the DAC update rule. It is shown that the developed scheme preserves the convergence properties of the conventional DAC framework while preventing information leakage to external eavesdroppers. Furthermore, the developed algorithm is applied to state-of-charge (SoC) balancing in a networked battery energy storage system to demonstrate its practical applicability. Simulation results validate the theoretical findings.
\end{abstract}


\vspace{-4mm}
\section{Introduction}
\label{sec:introduction}
\IEEEPARstart{A}{verage} consensus has become a fundamental tool in multi-agent systems. It has been widely adopted across various applications, such as sensor fusion~\cite{olfati2005consensus}, economic dispatch~\cite{chen2022privacy}, and distributed load sharing~\cite{meng2016distributed}, and multi-robot coordination~\cite{zhang2022privacy}. Consensus algorithms are typically categorized as either static or dynamic. Static consensus methods aim to compute the average of agents’ initial states~\cite{olfati2007consensus,ren2010distributed}, whereas dynamic average consensus (DAC) methods are designed to track the average of time-varying signals~\cite{bai2010robust,chen2012distributed,chen2014distributed}. More recently, DAC frameworks have gained increasing attention in practical control applications, such as the coordinated operation of networked battery energy storage systems (BESSs)~\cite{maithripala2026bess} and proportional current-sharing schemes while preserving voltage stability in DC microgrids~\cite{aguirre2025dynamic}.

In~\cite{freeman2006stability}, two DAC estimation schemes, termed the proportional method and the proportional-integral method, are introduced, each ensuring bounded steady-state tracking error. Earlier, the work in~\cite{spanos2005dynamic} proposed one of the first DAC algorithms capable of tracking the average of dynamic reference signals under suitable initial conditions. Finite-time convergence in DAC was investigated in~\cite{chen2012distributed} through the use of discontinuous update laws. To address initialization errors, robust DAC strategies with discontinuous dynamics were proposed in~\cite{george2019robust}.

The DAC methods discussed above generally rely on the exchange of local state information among neighboring agents. Such direct information sharing exposes the system to privacy and security risks as adversaries or eavesdroppers may intercept the transmitted signals. Access to these states enables attackers to infer sensitive quantities, such as local reference signals, and potentially exploit this knowledge to disrupt the operation of the network.
In distributed optimal economic dispatch, information exchange is required, which may expose sensitive data and enable adversaries to disrupt market operations or threaten microgrid reliability~\cite{chen2022privacy}.

Although various approaches have been developed to mitigate these risks. The majority of privacy-preserving methods are limited to static average consensus settings. Differential privacy~\cite{he2020differential,Nozari2015DifferentiallyPA,gao2018differentially} is widely used to protect sensitive information by adding noise to shared data. As a result, the states do not converge precisely to the true average. Studies in~\cite{mo2016privacy, chen2024new} show that carefully designed perturbations can still achieve accurate convergence. A cryptography-based approach is presented in~\cite{ruan2019secure} where transmitted data are encrypted to prevent eavesdropping. However, these methods introduce significant computational and communication costs due to the heavy processing required for encryption and decryption. To address this limitation, a state decomposition method with lower complexity was proposed in~\cite{wang2019privacy}.

The privacy-preserving schemes discussed above are developed for static average consensus problems. The work in~\cite{kia2015dynamic} develops a DAC algorithm that 
preserves the privacy of each agent’s reference signal against a privileged adversary. However, the privacy guarantee is valid only if the reference signals vary with time during an initial period. In addition, the state decomposition approach in~\cite{wang2019privacy} has been further developed for DAC in~\cite{zhang2022privacy}.

In this paper, we consider the DAC algorithm in~\cite{spanos2005dynamic},
\begin{equation}
\dot{\hat{z}}_{{\rm a},i}(t) \!=\! \dot{z}_i(t) - \beta\! \sum_{j=1}^{N}\!\! a_{ij}\big(\hat{z}_{{\rm a},i}(t) \!-\! \hat{z}_{{\rm a},j}(t)\big),\,
\hat{z}_{{\rm a},i}(0) \!=\! z_i(0) ,
\label{eq:1}
\end{equation}
where $z_i(t)$ is the time-varying reference signal of agent $i$, $\hat{z}_{{\rm a},i}(t)$ denotes agent $i$’s approximation of the average $z_{\rm a}(t)=\frac{1}{N}\sum_{i=1}^{N} z_i(t)$, $\beta>0$ is a tuning parameter, and $a_{ij}$ denotes the connectivity between agents $i$ and $j$. The lack of privacy protection in the DAC scheme~\eqref{eq:1} was identified in~\cite{zhang2022privacy}. The study demonstrated that both internal and external eavesdroppers can recover each agent’s private time-varying reference signal $z_i(t)$ and its derivative $\dot{z}_i(t)$ through observer-based methods. The privacy-preserving DAC scheme in~\cite{zhang2022privacy} protects against external eavesdroppers but exhibits slower convergence due to the modification of the graph Laplacian. Motivated by this observation, we proposed in~\cite{maithripala2026dac} a privacy-preserving DAC algorithm based on reference signal masking, which preserves the convergence rate and steady-state tracking accuracy of conventional DAC in~\eqref{eq:1} while safeguarding the reference signal of every agent against external eavesdroppers and maintaining lower computational complexity than the state decomposition method. In this approach, each agent masks its reference signal by adding a real-valued mask constructed from random values exchanged with its neighbors. However, since the mask remains constant after initialization, it protects only the reference signal and does not preserve the privacy of its derivative. Neither does it allow adjustment of the privacy level. 

To address this limitation, this study introduces a new masking strategy based on a sinusoidal masking signal. The proposed method maintains the convergence behavior of~\eqref{eq:1} while ensuring privacy protection for both the reference signals and their derivatives. Protecting both the reference signals and their derivatives is critical in applications such as formation control~\cite{zhang2022privacy} and state-of-charge (SoC) balancing~\cite{maithripala2026bess}, where they convey sensitive operational information. The proposed approach also allows the privacy level to be tuned through an adjustable design parameter.

In this paper, we provide a performance analysis demonstrating that the proposed scheme retains the original convergence rate and steady-state accuracy of \eqref{eq:1}. We further establish privacy against external eavesdropping by demonstrating that neither the reference signals nor their derivatives can be uniquely reconstructed from the available observations. Furthermore, a case study on SoC balancing of a networked BESS illustrates how the proposed method can be integrated with a power allocation law to realize privacy-preserving distributed control. Simulation results further validate the theoretical analysis.

The remainder of the paper is structured as follows. Section~\ref{sec:preliminaries} introduces the graph-theoretic concepts, while Section~\ref{sec:Main} describes the proposed privacy-preserving DAC scheme together with its convergence and privacy analyzes against external eavesdroppers. Section~\ref{sec:app} demonstrates its application to SoC balancing, where the algorithm is integrated with a distributed power allocation law. Section~\ref{sec:simu} provides simulation results validating the convergence and privacy properties. Finally, Section~\ref{sec:con} concludes the paper.

\section{Graph Theory}
\label{sec:preliminaries}
\subsection{Communication Network}
Let $\mathcal{G}(\mathcal{V}, \mathcal{E})$ be a connected undirected graph composed of $N$ agents, numbered as 1, 2, \ldots , $N$. Let $\mathcal{V} = \{1, 2, \ldots, N\}$ denote the set of agents and $\mathcal{E} \subseteq \mathcal{V} \times \mathcal{V}$ the set of edges. For each node $i$, define its neighborhood as $\mathcal{N}_i = \{j \in \mathcal{V} : (i, j) \in \mathcal{E}\}$, and $N_i = |\mathcal{N}_i|$ denotes the number of its neighbors.

The adjacency matrix of the graph $\mathcal{G}$ is $A = [a_{ij}]_{N \times N}$, where $a_{ij} > 0$ if $(j,i) \in \mathcal{E}$ and $a_{ij} = 0$  if $(j,i)\notin \mathcal{E}$. Let $L=(l_{ij}) \in \mathbb{R}^{N \times N}$ represent the Laplacian corresponding to the adjacency matrix $A$, where the entries satisfy $l_{ij} = -a_{ij}$ for $i \neq j$ and $l_{ii} = \sum_{j=1,\, j \neq i}^{N} a_{ij}$. 
Finally, $\mathbf{1}_n$ represents the all-ones vector in $\mathbb{R}^n$.

\begin{assumption}\label{ass: graph}
The communication graph is undirected and connected.
\end{assumption}

\begin{lemma}\label{lemma1}
\cite{kia2019tutorial} Let Assumption~\ref{ass: graph} hold and $\dot{z}_i(t)$ be bounded. The estimate $\hat{z}_{{\rm a},i}(t)$ generated by \eqref{eq:1} converges exponentially to a bounded neighborhood of $\frac{1}{N}\mathbf{1}_N z(t)$ for any $\beta > 0$, that is,
\begin{align*}
\lim_{t \to \infty} \sup \left| \hat{z}_{{\rm a},i}(t) - \tfrac{1}{N}\mathbf 1_Nz(t) \right| &\leq \frac{\gamma}{\beta \lambda_2},\quad i\in\mathcal V . 
\end{align*}
where
$
\gamma = \sup_{\tau \in [t,\infty)}
\left\|\left( I_N - \tfrac{1}{N}\mathbf{1}_N \mathbf{1}_N\tt \right)\dot{z}(\tau)\right\|, 
$
$z(t) = [z_1(t)\; z_2(t)\,\dotsc\, z_N(t)]\tt$, and $\lambda_2$ is the smallest nonzero eigenvalue of Laplacian matrix $L$.
\end{lemma}

\section{Main results}\label{sec:Main}

\subsection{Algorithmic Design for Privacy Protection}

The objective of this study is to develop a privacy-preserving DAC algorithm that prevents external eavesdroppers from inferring each agent’s time-varying reference signal $z_i(t)$ and its derivative $\dot{z}_i(t)$. We consider an adversary monitoring all information exchanged among neighboring agents and attempting to reconstruct the agents’ sensitive states from these observations. To address this problem, we propose a privacy-preserving DAC algorithm based on reference-signal masking.
The algorithm begins with reference signal masking, in which each agent $i$ generates pairwise masking signals $s_{ij}(t)$ for every neighbor $j \in \mathcal{N}_i$ and transmits $s_{ij}(t)$ to agent $j$ through a secure communication link. The signals are defined as
\begin{equation}
s_{ij}(t) = A_{\rm m}\sin\!\big(\omega_{ij} t \big),
\quad j \in \mathcal{N}_i,
\end{equation}
where $\omega_{ij}$ denotes a random constant frequency generated by agent $i$ for agent $j$, and $A_{\rm m}$ denotes a common predefined amplitude shared by all agents. Using the received pairwise signals, agent $i$ then constructs its masking signal $m_i(t)$ according to the following expression.
 \begin{equation}
  m_i(t) = \sum_{j\in\mathcal{N}_i}\big(s_{ji}(t)-s_{ij}(t)\big),
 \label{eq:3}
 \end{equation}
and this construction guarantees that
$
\sum_{i=1}^N m_i(t)=0,
\;
\sum_{i=1}^N \dot m_i(t) = 0,
$
for all $t \ge 0$. The masked reference signal $z_{{\rm m}, i}(t)$ is defined as
\begin{equation}
    z_{{\rm m}, i}(t) = z_i(t) + m_i(t),
    \label{eq:4}
\end{equation}
thereby preserving the global average while concealing the individual reference signals and their dynamics in~\eqref{eq:1}. Finally, the estimator dynamics for each agent take the form
\begin{equation}
\begin{aligned}
\dot{\hat{z}}_{{\rm a},i}(t) &= \dot{z}_{\rm m,i}(t) - \beta \sum_{j=1}^{N} a_{ij}\big(\hat{z}_{{\rm a},i}(t) - \hat{z}_{{\rm a},j}(t)\big),\\
\hat{z}_{{\rm a},i}(0) &= z_i(0)+m_i(0) ,\ \ i\in\mathcal{V}.
\end{aligned}
\label{eq:5}
\end{equation}

\begin{assumption}\label{ass:encryption}
All communications between neighboring agents are encrypted during masking generation to prevent information leakage.

Since this exchange is brief and occurs only during initialization, the associated computational overhead is negligible~\cite{chen2024new}.
\end{assumption}


\subsection{Convergence Analysis}

This subsection analyzes the convergence properties of the developed algorithm.

\begin{theorem}
\label{thm:ppdac-sinusoidal}
Let Assumption~\ref{ass: graph} hold. Each agent constructs the masking signal~\eqref{eq:3}, forms the masked reference signal~\eqref{eq:4}, and implements the DAC update~\eqref{eq:5}. For any reference signals $z_i(t)$ with bounded derivatives, the estimates $\hat{z}_{{\rm a},i}(t)$ track the network true average $\tfrac{1}{N}\mathbf{1}_N^{\rm T} z(t)$ with a steady-state deviation that remains bounded and decreases as $\beta$ increases.
\end{theorem}

\begin{proof}
Define the matrix $\Psi(t)=[s_{ij}(t)]\in\mathbb R^{N\times N}$, where $s_{ij}(t)=0$ if $j\notin\mathcal N_i$ or $j=i$. 
Then, the masking signal vector can be expressed as
$
m(t) = \big(\Psi(t)^{\rm T}-\Psi(t)\big)\mathbf 1_N,
$ where $m(t) = [m_1(t)\; m_2(t)\, \dotsc\, m_N(t)]^{\rm T}$.
By construction,
\[
\mathbf 1_N^{\rm T} m(t)
= \mathbf 1_N^{\rm T}\Psi(t)^{\rm T}\mathbf 1_N
- \mathbf 1_N^{\rm T}\Psi(t)\mathbf 1_N
= 0
\quad t\ge0.
\]
Since $s_{ij}(t)$ are differentiable functions of time, it also follows that
$
\mathbf 1_N^{\rm T} \dot m(t)=0, \; t\ge0.
$

Let $z_{\rm m}(t)=[z_{{\rm m},1}(t)\;z_{{\rm m},2}(t)\;\dots\;z_{{\rm m},N}(t)]^{\rm T}$ and $z(t)=[z_1(t)\; z_2(t)\ \dotsc\; z_N(t)]^{\rm T}$. Then,
$
\mathbf 1_N^{\rm T} z_{\rm m}(t)
= \mathbf 1_N^{\rm T} z(t)+\mathbf 1_N^{\rm T} m(t)
= \mathbf 1_N^{\rm T} z(t),
$
so the network average of the masked reference equals the true average for all $t$. Differentiating $z_{{\rm m},i}(t)$ yields
$
\dot z_{{\rm m},i}(t)=\dot z_i(t)+\dot m_i(t),
$
and expressing the DAC update~\eqref{eq:5} in compact vector form gives
$
\dot{\hat z}_{\rm a}(t)
= \dot z(t)+\dot m(t)-\beta L\hat z_{\rm a}(t),
$ where $\hat{z}_{{\rm a},i}(t)=[\hat{z}_{{\rm a},1}\; \hat{z}_{{\rm a},2}\, \dotsc\, \hat{z}_{{\rm a},N}]^{\rm T}$. This system coincides with the conventional DAC dynamics driven by the inputs $\dot z_i(t)$ and an additional disturbance terms $\dot m_i(t)$. Therefore, the standard DAC result in Lemma~\ref{lemma1} implies that the tracking error of the designed algorithm satisfies
\[
\limsup_{t\to\infty}
\left|\hat z_{{\rm a},i}(t)-\tfrac{1}{N}\mathbf 1_N^{\rm T} z(t)\right|
\le \frac{\gamma_s}{\beta\lambda_2},
\quad i\in\mathcal V,
\]
where
$
\gamma_s
= \sup_{\tau\ge t}
\left\|
\Big(I_N-\tfrac{1}{N}\mathbf 1_N\mathbf 1_N^{\rm T}\Big)
\big(\dot z(\tau)+\dot m(\tau)\big)
\right\|,
$
and $\lambda_2$ denotes the second-smallest eigenvalue of $L$.

Because $\dot m(t)$ is bounded, the error bound can be made arbitrarily small by selecting the gain $\beta$ sufficiently large. This completes the proof.
\end{proof}

\begin{remark}
The convergence speed of consensus-based algorithms is characterized by $\lambda_2$ the smallest nonzero eigenvalue of Laplacian matrix $L$~\cite{olfati2007consensus}. For dynamic average consensus, the exponential decay rate of the error is governed by $\beta\lambda_2$~\cite{kia2019tutorial}. As shown in Theorem~1, the proposed reference-masking mechanism does not modify the Laplacian matrix. Consequently, the developed scheme retains the exponential convergence rate of~\eqref{eq:1}.
\end{remark}

\vspace{-3mm}
\subsection{Privacy Analysis }

We examine the privacy guarantees of the proposed scheme under the presence of an external eavesdropper. Assuming that the eavesdropper has access to all communication link in the network, the information set available to it can be expressed as
$
I_{\rm obs}=\{\, A,\ \beta,\ \hat{z}_{{\rm a},i}(t)\mid i\in\mathcal V,\ t\ge 0 \,\}.
$

\begin{theorem}
\label{thm:ext_privacy}
Let Assumptions~\ref{ass: graph} and~\ref{ass:encryption} hold. Consider the developed privacy-preserving DAC scheme in which each agent uses the masked reference signal \eqref{eq:4}, with the node mask constructed as in \eqref{eq:3} and runs the DAC update~\eqref{eq:5}. Information available to an external eavesdropper is $I_{\rm obs}$.
Then the external eavesdropper cannot uniquely identify any individual agent's reference trajectory $z_i(t)$ or its derivative $\dot{z}_i(t)$ from its observations.
\end{theorem}

\begin{proof}
Let $z(t)=[z_1(t)\; z_2(t)\, \dotsc\, z_N(t)]^{\rm T}$, $m(t) = [m_1(t)\; m_2(t)\, \dotsc\, m_N(t)]^{\rm T}$, and $\hat{z}_{{\rm a},i}(t)=[\hat{z}_{{\rm a},1}\; \hat{z}_{{\rm a},2}\, \dotsc\, \hat{z}_{{\rm a},N}]^{\rm T} $. Masks construction guarantees that
$
\label{eq:mask_zero_sum_ext}
\mathbf 1_N^{\rm T}\ m(t)=0,
\qquad
\mathbf 1_N^{\rm T} \dot m(t)=0
$, for all $t\ge 0$.

Define the stacked masked reference signals $z_{\rm m}(t)=z(t)+m(t)$. Then
$
\mathbf 1_N^{\rm T} z_{\rm m}(t)
=\mathbf 1_N^{\rm T} z(t),
$ and
$
\mathbf 1_N^{\rm T} \dot z_{\rm m}(t)
=\mathbf 1_N^{\rm T} \dot z(t).
$

The stacked DAC dynamics driven by the masked references can be written as
\begin{equation}
\label{eq:dac_stack_ext}
\dot{\hat z}_{\rm a}(t)
=\dot z_{\rm m}(t)-\beta L\hat z_{\rm a}(t)
=\dot z(t)+\dot m(t)-\beta L\hat z_{\rm a}(t).
\end{equation}
The external eavesdropper observes $\hat z_{\rm a}(t)$ for all $t\ge 0$ and knows $L$ and $\beta$, but does not observe $m(t)$, $\dot m(t)$, or secret signals $s_{ij}(t)$.

We prove non-identifiability by constructing infinitely many distinct reference trajectories that generate the same observed estimator trajectory.
Let $\delta(t)=[\delta_1(t)\; \delta_2(t)\, \dotsc\, \delta_N(t)]^{\rm T}\in\mathbb R^N$ be any continuously differentiable signals satisfying
\begin{equation}
\label{eq:delta_zero_sum_ext}
\mathbf 1_N^{\rm T} \delta(t)=0,
\qquad
\mathbf 1_N^{\rm T} \dot\delta(t)=0,
\qquad  t\ge 0.
\end{equation}
Define an alternative execution by
$
z'(t)=z(t)+\delta(t),\;
 m'(t)=m(t)-\delta(t).
$
It follows that
$
\mathbf 1_N^{\rm T}m'(t)
=\mathbf 1_N^{\rm T} m(t)-\mathbf 1_N^{\rm T} \delta(t)
=0,
\;
\mathbf 1_N^{\rm T} \dot m'(t)=0,
$
so $m'(t)$ satisfies the same zero-sum property as the original mask $m(t)$ and is therefore admissible under the same construction.

Moreover, the masked references and their derivatives are unchanged:
$
z'(t)+m'(t)=z(t)+m(t)=z_{\rm m}(t),
$
$
\dot z'(t)+\dot m'(t)
=\dot z(t)+\dot m(t)
=\dot z_{\rm m}(t).
$
Consequently, the DAC dynamics~\eqref{eq:dac_stack_ext} driven by $(z,m)$ and by $(z',m')$ are identical and generate exactly the same estimator trajectory $\hat z_{\rm a}(t)$ for all $t\ge 0$.

Since $\delta(t)$ can be chosen arbitrarily (subject to the zero-sum constraint~\eqref{eq:delta_zero_sum_ext}), for any agent $i$ there exist infinitely many distinct trajectories $\delta_i(t)$ satisfying
$
\mathbf 1_N^{\rm T} \delta(t)=0, \quad \forall t\ge 0,
$
such that
$
z_i'(t)=z_i(t)+\delta_i(t), \;
\dot z_i'(t)=\dot z_i(t)+\dot\delta_i(t),
$
which produce the same observed estimator $\hat z_{\rm a}(t)$.

Therefore, neither $z_i(t)$ nor its derivative $\dot z_i(t)$ can be uniquely identified from the information available to the external eavesdropper.
This completes the proof.
\end{proof}

\section{Application to SoC balancing for Battery Energy Storage Systems}\label{sec:app}

In this section, we apply the privacy-preserving DAC algorithm to the SoC balancing and power tracking problems in a networked BESS, as shown in Fig.~\ref{fig:MG}.

\vspace{-4mm}
\subsection{The SoC Dynamics and Control Objective}

Consider a BESS composed of $N$ battery units. For each unit $i \in \mathcal V$, the SoC $S_i(t)$ satisfies the Coulomb counting model
\vspace{-2mm}
\begin{equation}\label{eq:9}
S_i(t)=S_i(0)-\frac{1}{C_i}\int_0^t i_i(\tau)\,d\tau,
\end{equation}
where $C_i$ is the battery capacity and $i_i(t)$ is the output current with $i_i(t) > 0$ denoting discharging and $i_i(t) < 0$ charging. Differentiating \eqref{eq:9} yields
\vspace{-1mm}
\begin{equation}\label{eq:10}
\dot S_i(t)=-\frac{1}{C_i}i_i(t).
\end{equation}

The power delivered by unit $i$ is given by 
$
p_i(t)=V_i i_i(t),
$
where $V_i$ represents the DC--DC converter terminal voltage, which is assumed to be constant. Hence,
\vspace{-2mm}
\begin{equation}
\dot S_i(t)=-\frac{1}{C_i V_i}p_i(t), \quad i \in \mathcal V.
\end{equation}

Our objective is to design a privay-preserving distributed power allocation law using only neighbor-to-neighbor communication to achieve SoC balancing and total power tracking.

\textbf{Problem 1:} We study a BESS composed of $N$ battery units whose SoC dynamics are described by~\eqref{eq:10}. Design a privacy-preserving distributed power allocation law such that

(i) the SoC levels of the battery units are balanced, that is,
\begin{equation}
\lim_{t\to\infty}|S_i(t)-S_j(t)|\le\varepsilon_S,
\quad  i,j \in \mathcal V,
\end{equation}

(ii) the desired total power is tracked by the total power output of the battery units, satisfying
\begin{equation}
\lim_{t\to\infty}\left|\sum_{i=1}^N p_i(t)-p^*(t)\right|
\le\varepsilon_P,
\end{equation}
where $\varepsilon_S>0$ and $\varepsilon_P>0$ denote prescribed accuracy levels.

\begin{assumption}\label{ass:power}
$p^*(t)$ and its time derivative are bounded such that
\vspace{-2mm}
\begin{equation}
p_{\rm min} \le |p^*(t)| \le p_{\rm max},\quad
|\dot{p}^*(t)| \le \psi,\quad t \ge 0,
\end{equation}
where $p_{\rm min}>0$, $ p_{\rm max}>0$, and $\psi>0$ are constants.
\end{assumption}

\vspace{1mm}
\begin{assumption}\label{ass:access}
The desired total power $p^*(t)$ is available to at least one battery unit.
\end{assumption}

\subsection{Power Allocation and Average Estimation}
The state of battery unit $i$ is defined according to the operating mode as
\begin{equation}
x_i(t)=
\begin{cases}
x_{{\rm d},i}(t)=C_i V_i S_i(t), & \text{during discharging},\\
x_{{\rm c},i}(t)=C_i V_i\big(1-S_i(t)\big), & \text{during charging}.
\end{cases}
\end{equation}

The variable $x_i(t)$ represents the amount of energy that can be stored during the charging phase or the available electrical energy that can be discharged during the discharging phase. Physical constraints imply the existence of constants $a_1>0$ and $a_2>0$ satisfying
$
a_1 \le x_i(t) \le a_2, t\ge0, i\in \mathcal V.
$

The corresponding energy evolution satisfies
\begin{align}
\dot{x}_{{\rm d},i}(t) &= C_i V_i \dot{S}_i(t) = -p_i(t) \quad \text{(discharging)}, \\
\dot{x}_{{\rm c},i}(t) &= -C_i V_i \dot{S}_i(t) = p_i(t) \quad \text{(charging)}.
\end{align}

Let $p^*(t)$ be the total desired power. The local power allocation rule is
\begin{equation}
p_i(t)=
\begin{cases}
\dfrac{x_{{\rm d},i}(t)}{\sum_{j=1}^{N}x_{{\rm d},j}(t)}\,p^*(t), & \text{discharging},\\
\dfrac{x_{{\rm c},i}(t)}{\sum_{j=1}^{N}x_{{\rm c},j}(t)}\,p^*(t), & \text{charging}.
\end{cases}
\end{equation}

Define the network average unit state
\begin{equation}
x_{\rm a}(t)=
\begin{cases}
\dfrac{1}{N}\sum_{i=1}^{N}x_{{\rm d},i}(t), & \text{discharging},\\
\dfrac{1}{N}\sum_{i=1}^{N}x_{{\rm c},i}(t), & \text{charging},
\end{cases}
\end{equation}
and the average desired power
$
p_{\rm a}(t)=\dfrac{1}{N}p^*(t).
$
Then the power allocation can be written compactly as
\begin{equation}
p_i(t)=\frac{x_i(t)}{x_{\rm a}(t)}\,p_{\rm a}(t).
\end{equation}

Because both $x_{\rm a}(t)$ and $p_{\rm a}(t)$ depend on global information, they are obtained through distributed estimation. To estimate the average unit state $x_{\rm a}(t)$, we utilize the designed privacy-preserving DAC scheme in Section~\ref{sec:Main} as follows,
\begin{equation}
\begin{aligned}
\dot{\hat{x}}_{{\rm a},i}(t) &= \dot{x}_{\rm m,i}(t) - \beta \sum_{j=1}^{N} a_{ij}\big(\hat{x}_{{\rm a},i}(t) - \hat{x}_{{\rm a},j}(t)\big),\\
\hat{x}_{{\rm a},i}(0) &= x_i(0)+m_i(0) ,\ \ i\in\mathcal{V},
\end{aligned}
\label{eq:PP-DAC}
\end{equation}
where $\hat{x}_{{\rm a},i}(t)$ is the estimate of $x_{\rm a}(t)$. 
According to Theorem~\ref{thm:ppdac-sinusoidal}, $\hat{x}_{{\rm a},i}(t)$ approaches a bounded neighborhood of $\frac{1}{N}\sum_{i=1}^{N} x_i(t)$. Furthermore, Theorem~\ref{thm:ext_privacy} shows that the individual battery unit’s state $x_i(t)$ and its derivative ($\dot{x}_i(t) = -p_i(t)$ during discharging and $\dot{x}_i(t) = p_i(t)$ during charging) remain inaccessible to any external eavesdropper.

To estimate the average desired power, each unit implements
\begin{equation*}
\dot{\hat{p}}_{{\rm a},i}(t)\!=\!-\kappa \!\left( \sum_{j=1}^{N} a_{ij} (\hat{p}_{{\rm a},i}(t) - \hat{p}_{{\rm a},j}(t))\! + b_i (\hat{p}_{{\rm a},i}(t) - p_{\rm a}(t))\! \!\right)\!\!,
\end{equation*}
\begin{equation}\label{eq:p_estimator}
\hat{p}_{{\rm a},i}(0) = 0,
\end{equation}
where $\hat{p}_{{\rm a},i}(t)$ is the estimate of $p_{\rm a}(t)$, $\kappa>0$ is a tuning parameter, and $b_i\ge0$ indicates whether unit $i$ has direct access to $p_{\rm a}(t)$. Using these estimates, the implementable distributed power law becomes
\begin{equation}\label{eq:power allocation}
p_i(t)=
\begin{cases}
\dfrac{x_{{\rm d},i}(t)}{\max\!\left\{\frac{a_1}{2},\,\hat{x}_{{\rm a},i}(t)\right\}}
\,\hat{p}_{{\rm a},i}(t), & \text{discharging},\\
\dfrac{x_{{\rm c},i}(t)}{\max\!\left\{\frac{a_1}{2},\,\hat{x}_{{\rm a},i}(t)\right\}}
\,\hat{p}_{{\rm a},i}(t), & \text{charging}.
\end{cases}
\end{equation}

\begin{lemma}\cite{cai2016distributed}
\label{lem:2}
Let Assumptions~\ref{ass: graph}, \ref{ass:power}, and \ref{ass:access} hold. 
For any $\kappa > 0$, the estimator $\hat{p}_{{\rm a},i}(t)$ defined in~\eqref{eq:p_estimator} approaches a bounded neighborhood of $p_{\rm a}(t)$ at an exponential rate. Moreover, the steady-state estimation error satisfies
\[
\limsup_{t \to \infty} 
\left| \hat{p}_{{\rm a},i}(t) - p_{\rm a}(t) \right|
\le \frac{\psi \gamma_{\rm p}}{\kappa},
\]
where $\gamma_{\rm p} > 0$ is a constant.
\end{lemma}

The following theorem characterizes the performance of the proposed distributed power allocation scheme.
\begin{theorem}\label{thm:3}
Suppose that the battery parameters $C_i$ and $V_i$ are known, and 
Assumptions~\ref{ass: graph}--\ref{ass:access} hold. Then, the distributed power allocation law~\eqref{eq:power allocation}, 
together with the privacy-preserving average unit state estimator~\eqref{eq:PP-DAC}, and average desired power estimator~\eqref{eq:p_estimator}, 
achieves the objective stated in Problem~1.
\end{theorem}

\begin{proof}
We present the proof only for the discharging case. The estimation errors are defined as
$
 e_{x,i}(t)=\hat x_{{\rm a},i}(t)-x_{\rm a}(t),\;
 e_{p,i}(t)=\hat p_{{\rm a},i}(t)-p_{\rm a}(t).
$
From the  Theorem~\ref{thm:ppdac-sinusoidal} and Lemma~\ref{lem:2} , the estimator errors are ultimately bounded and their bounds decrease with $\beta$ and $\kappa$. Choose \( \beta \geq \frac{2\gamma_{\rm s}}{a_1 \lambda_2} \) so that there exists $T>0$ with
$\hat x_{{\rm a},i}(t)\ge \tfrac12 a_1$ for all $t\ge T$.

For $t\ge T$, the discharging power can be written as
\begin{align*}
p_i(t)
&=\frac{x_{{\rm d},i}(t)}{x_{\rm a}(t)+ e_{x,i}(t)}\bigl(p_{\rm a}(t)+ e_{p,i}(t)\bigr)\\
&=\frac{\left(1+\frac{e_{p,i}(t)}{p_{\rm a}(t)}\right)}
       {\left(1+\frac{ e_{x,i}(t)}{x_{\rm a}(t)}\right)}
   \frac{p_{\rm a}(t)}{x_{\rm a}(t)}\,x_{{\rm d},i}(t) \\
&=k(t)\bigl(1+\Delta_i(t)\bigr)x_{{\rm d},i}(t),
\end{align*}
where
\[
k(t)=\frac{p_{\rm a}(t)}{x_{\rm a}(t)},\;
\Delta_i(t)=
\frac{1+\frac{ e_{p,i}(t)}{p_{\rm a}(t)}}
     {1+\frac{ e_{x,i}(t)}{x_{\rm a}(t)}}-1.
\]

Recalling $\dot S_i(t)=-\frac{1}{C_iV_i}p_i(t)$ and $x_{{\rm d},i}(t)=C_iV_iS_i(t)$,
\[
\dot S_i(t)=-k(t)\bigl(1+\Delta_i(t)\bigr)S_i(t). 
\]

Consider $V_{ij}=\tfrac12(S_i-S_j)^2$. Then its time derivative is given by
$
\dot V_{ij}
=-k(t)(S_i-S_j)\bigl((1+\Delta_i)S_i-(1+\Delta_j)S_j\bigr).
$
Thus $\dot V_{ij}<0$ whenever $S_i>S_j$ and $\frac{S_i}{S_j}>\frac{1+\Delta_j}{1+\Delta_i}$ (and vice versa). Since $\Delta_i(t)$ are bounded and can be made arbitrarily small by choosing $\beta,\kappa$ large.
In steady state there exist $\Delta^- = \frac{1-\frac{N\psi\gamma_p}{\kappa p_{\rm min}}}
     {1+\frac{ \gamma_s}{a_1\beta \lambda_2}}-1<0 <\Delta^+ =  \frac{1+\frac{N\psi\gamma_p}{\kappa p_{\rm min}}}
     {1-\frac{ \gamma_s}{a_1\beta \lambda_2}}$ such that
\[
\frac{1+\Delta^-}{1+\Delta^+}\le \frac{S_i}{S_j}\le
\frac{1+\Delta^+}{1+\Delta^-}.
\]
Hence SoC balancing is achieved with any prescribed accuracy $\varepsilon_s$ by selecting $\beta,\kappa$ sufficiently large. The total discharging power of the BESS is
\vspace{-2mm}
\[
p_\Sigma(t)=\sum_{i=1}^N p_i(t)
          =k(t)\sum_{i=1}^N (1+\Delta_i(t))x_{{\rm d},i}(t).
\]
\vspace{-3mm}
Since $\sum_{i=1}^N x_{{\rm d},i}(t)=N x_{\rm a}(t)$ and $k(t)=\frac{p_{\rm a}(t)}{x_{\rm a}(t)}$, we get
\vspace{-0.5mm}
\begin{align*}
p_\Sigma(t)
&=p^\ast(t)+k(t)\sum_{i=1}^N \Delta_i(t)x_{{\rm d},i}(t).
\end{align*}
\vspace{-1mm}
If $\Delta^-\le \Delta_i(t)\le \Delta^+$ for all $i$ in steady state, then
\vspace{-1mm}
\[
\Delta^- \sum_{i=1}^N x_{{\rm d},i}(t)\le \sum_{i=1}^N \Delta_i(t)x_{{\rm d},i}(t)
\le \Delta^+ \sum_{i=1}^N x_{{\rm d},i}(t),
\]
which implies
\[
(1+\Delta^-)p^\ast(t)\le p_\Sigma(t)\le (1+\Delta^+)p^\ast(t).
\]
Hence, choosing sufficiently large $\beta$ and $\kappa$ allows $p_\Sigma(t)$ to become arbitrarily close to $p^\ast(t)$.
\end{proof}
\vspace{-4mm}
\section{Simulation Results}\label{sec:simu}
\begin{figure}[!t]
\centering
\includegraphics[width=\columnwidth]{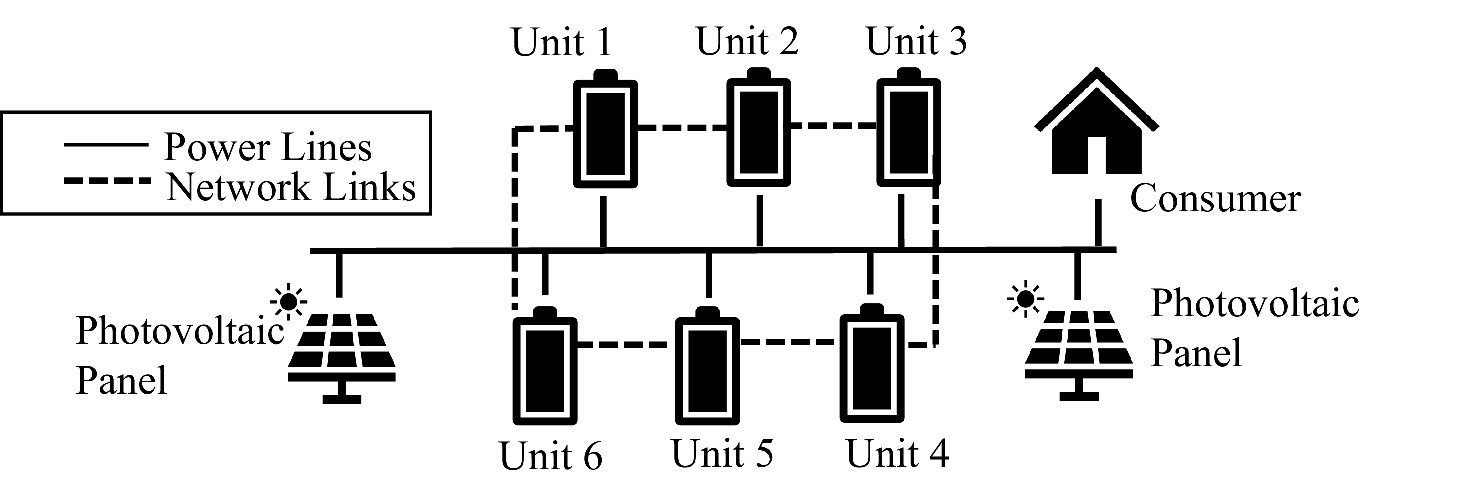}
\caption{Configuration of a six-unit networked BESS.}
\label{fig:MG}
\end{figure}
We study a six-unit battery energy storage system interconnected in a ring topology, as illustrated in Fig.~\ref{fig:MG}. Only the first unit has access to the desired power signal $p^{*}(t) = 4200\sin(t) + 4200~W$. The battery capacities are given by $(C_1, C_2, C_3, C_4, C_5, C_6) = (180, 190, 200, 210, 220, 230)$Ah., and all units have a nominal voltage of $50$~V. The initial state of charge is selected as $(S_1(0), S_2(0), S_3(0), S_4(0), S_5(0), S_6(0)) = (0.96, 0.89, 0.75, 0.80, 0.73, 0.88)$. For the distributed estimators, the design parameters are chosen as $\beta = 400$, and $\kappa = 300$.

We set the signal amplitude as $A_{\rm m}=500$ to implement the privacy-preserving DAC. Each agent assigns independent random frequency values $\omega_{ij}$ to its neighboring agents and these parameters are represented in the matrix
\[
\omega =
\begin{bmatrix}
0 & 1.11 & 0 & 0 & 0 & 3.37 \\
6.12 & 0 & 2.46 & 0 & 0 & 0 \\
0 & 4.03 & 0 & 3.80 & 0 & 0 \\
0 & 0 & 8.15 & 0 & 2.49 & 0 \\
0 & 0 & 0 & 5.75 & 0 & 6.89 \\
5.22 & 0 & 0 & 0 & 6.42 & 0
\end{bmatrix}.
\]
The $(i,j)$ entry of $\omega$ represents the random frequency selected by agent $i$ for agent $j$ to generate the signal $s_{ij}(t)=500\sin\big(\omega_{ij} t\big)$ for neighboring agents, which is then used to construct the  masking signal $m_i(t)$.

Fig.~\ref{fig:soc_power} shows that the six battery units reach SoC balancing while the total power output accurately follows the desired discharging power. Fig.~\ref{fig:estimators} shows that the privacy-preserving average unit state estimator and the estimated average desired power accurately track their corresponding true averages.

Theorem~\ref{thm:ext_privacy} guarantees that an external observer cannot reconstruct the private states $x_i(t)$ or their derivatives $\dot{x}_i(t)$. To illustrate this, we consider a discharging scenario where an eavesdropper attempts to estimate $x_i(t)$ and the individual power outputs $p_i(t)=-\dot{x}_i(t)$ using the observer in \cite{zhang2022privacy}. As shown in Figs.~\ref{fig:x_each_obs} and~\ref{fig:power_each_obs}, when the proposed privacy mechanism is implemented, the reconstructed signals deviate noticeably from the true trajectories $x_i(t)$ and $p_i(t)$, confirming the effectiveness of the privacy protection. The privacy of each battery unit’s power signal $p_i(t)=-\dot{x}_i(t)$ is measured using the root-mean-square error (RMSE). This is obtained by comparing the attacker’s estimate with the true signal and computing the RMS of the difference after the transient period. Fig.~\ref{fig:privacy} shows the privacy metric evaluated for different signal amplitudes $A_{\rm m}$, indicating that larger values lead to stronger privacy and confirming its role as a tuning parameter.
\begin{figure}[t]
\centering
\newcommand{\figH}{2.4cm}
\begin{minipage}[t]{0.49\columnwidth}
  \centering
  \includegraphics[height=\figH]{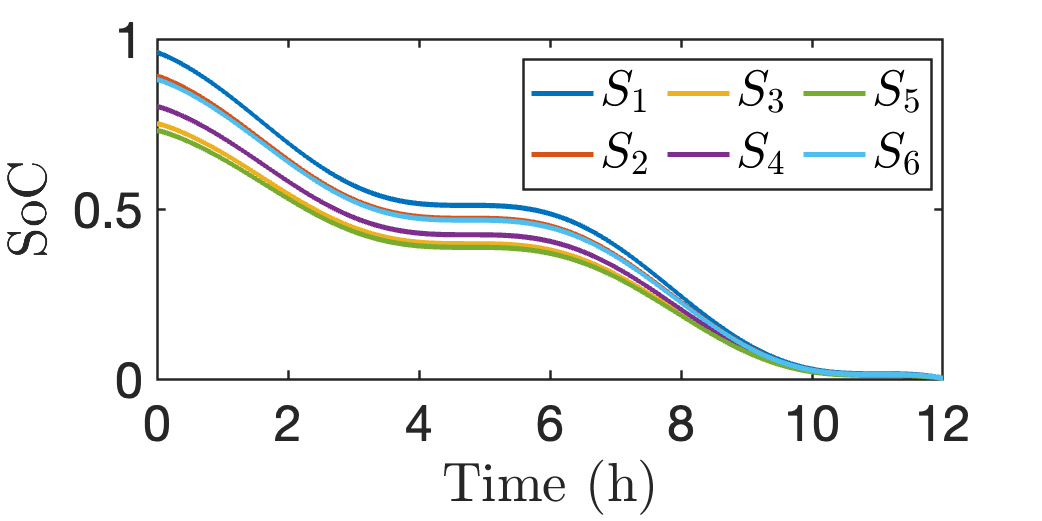}\\[-1mm]
  {\footnotesize (a) SoC balancing}
\end{minipage}\hfill
\begin{minipage}[t]{0.49\columnwidth}
  \centering
  \includegraphics[height=\figH]{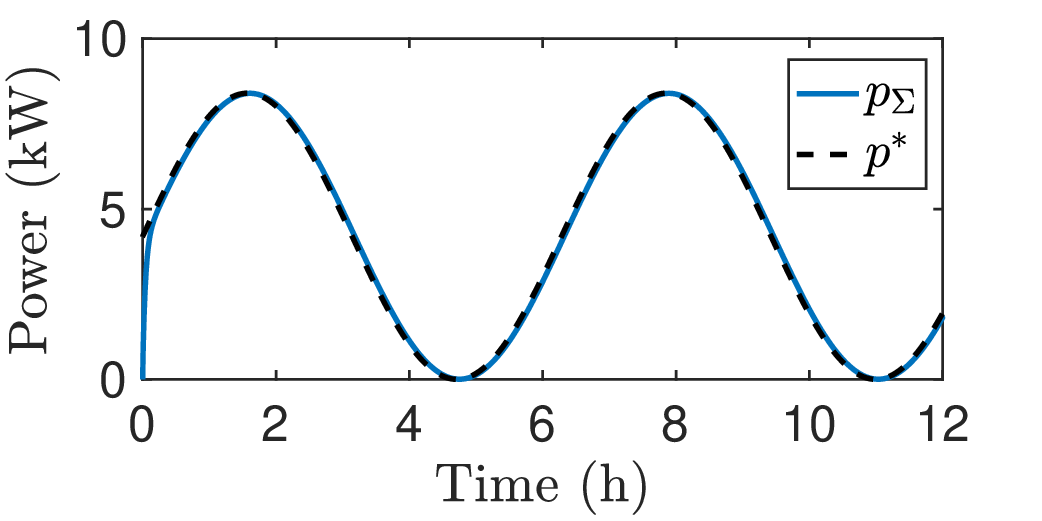}\\[-1mm]
  {\footnotesize (b) Power tracking}
\end{minipage}
\caption{State-of-charge balancing and power tracking performance.}
\label{fig:soc_power}
\end{figure}

\begin{figure}[t]
\centering
\newcommand{\figH}{2.35cm}
\begin{minipage}[t]{0.5\columnwidth}
  \centering
  \includegraphics[height=\figH]{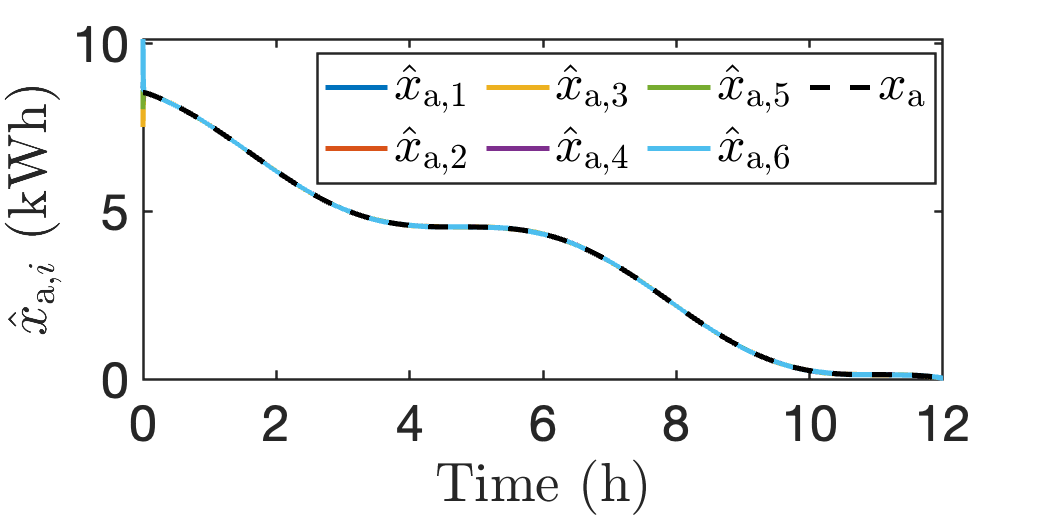}\\[-1mm]
  {\footnotesize (a)Estimation of the average unit state under privacy protection.}
\end{minipage}\hfill
\begin{minipage}[t]{0.5\columnwidth}
  \centering
  \includegraphics[height=\figH]{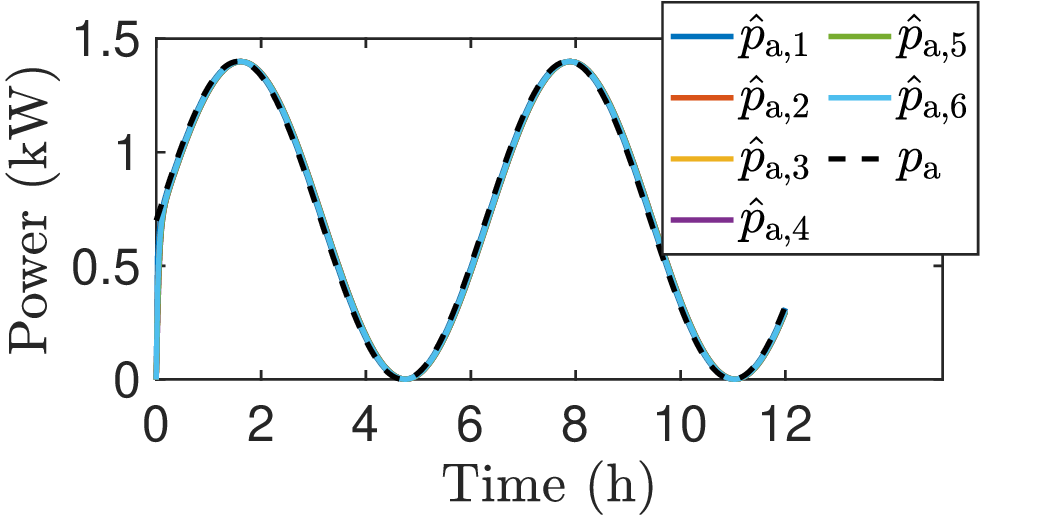}\\[-1mm]
  {\footnotesize (b) Estimation of average desired power.}
\end{minipage}
\caption{Estimator performance.}
\label{fig:estimators}
\end{figure}

\begin{figure}[!t]
\centering
\includegraphics[width=\columnwidth]{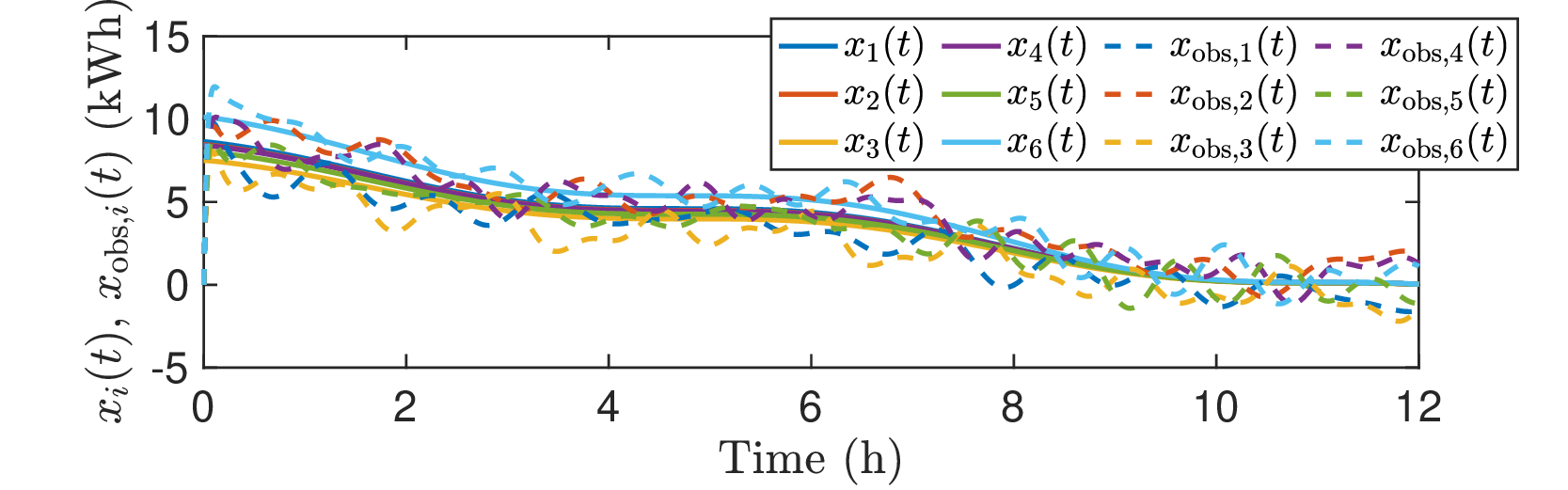}
\caption{Eavesdropper estimates $x_{{\rm obs},i}(t)$ vs true states $x_i(t)$}
\label{fig:x_each_obs}
\end{figure}

\begin{figure}[!t]
\centering
\includegraphics[width=\columnwidth]{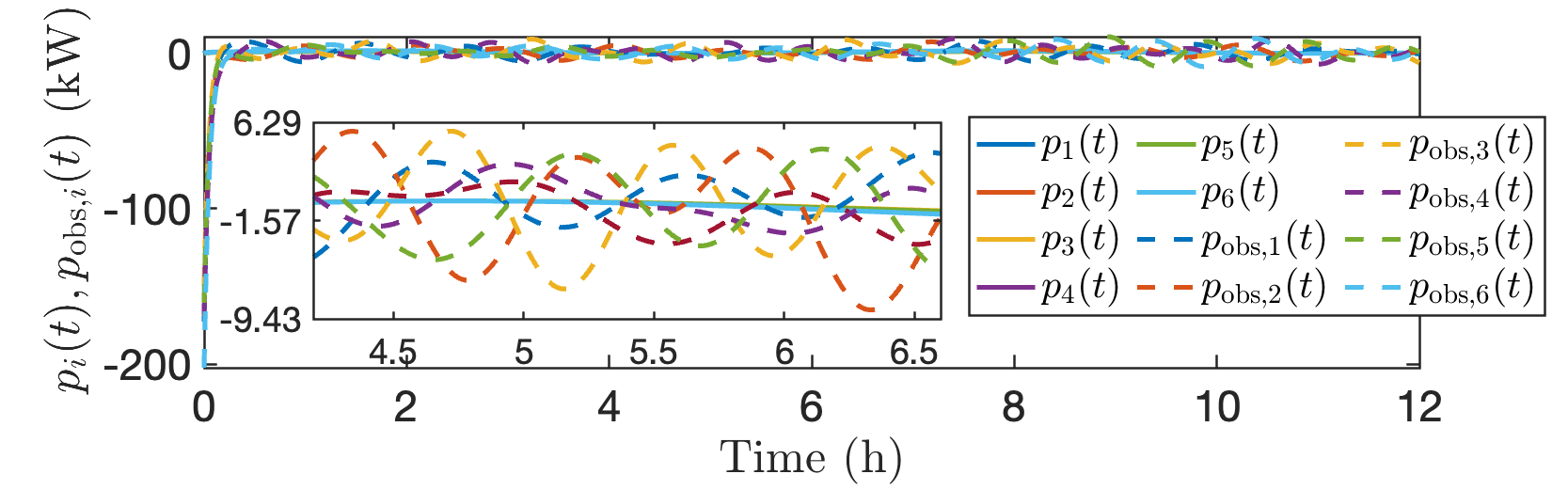}
\caption{Eavesdropper estimates $p_{{\rm obs},i}$ vs true states $p_i(t)$.}
\label{fig:power_each_obs}
\end{figure}

\begin{figure}[!t]
\centering
\includegraphics[width=\columnwidth]{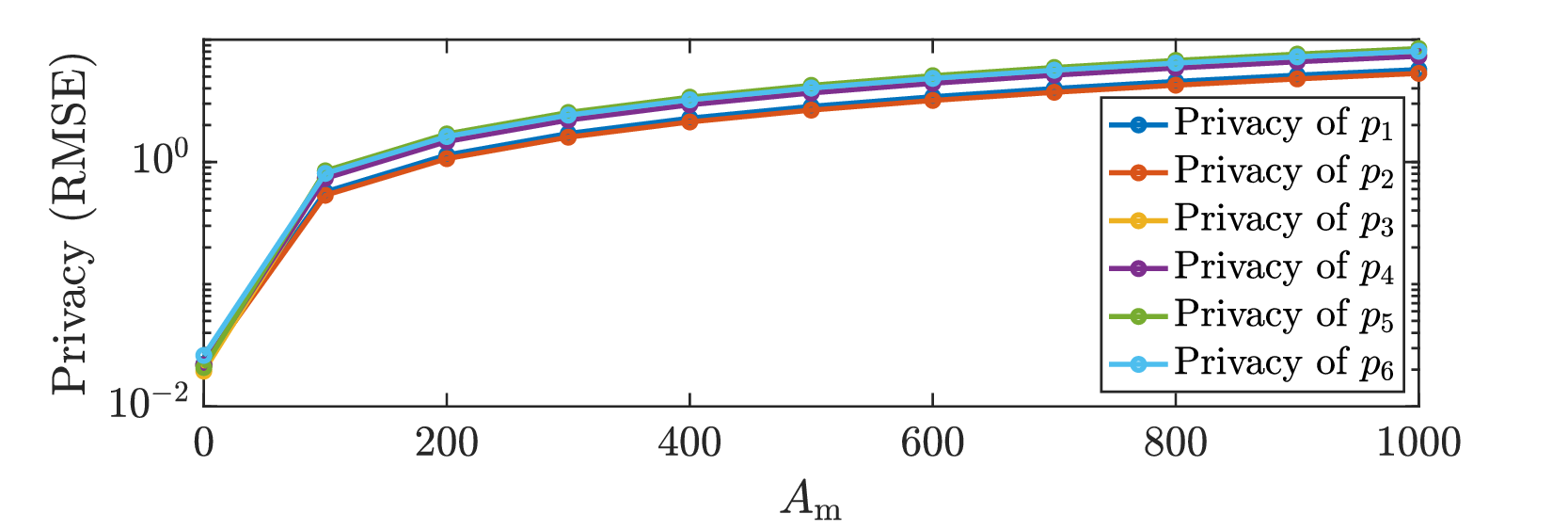}
\caption{Effect of $A_{\rm m}$ on the privacy of $p_i$}
\label{fig:privacy}
\end{figure}

\section{Conclusions}\label{sec:con}
In this paper, we developed a privacy-preserving DAC algorithm based on reference signal masking. We proved that the proposed algorithm preserves the privacy of both the reference signals and their derivatives without compromising the convergence properties of the conventional DAC. Furthermore, case studies were presented to demonstrate the practicality of the proposed approach in SoC balancing for a networked BESS. Simulation results illustrated effective SoC balancing and accurate power tracking while preventing the disclosure of private information of individual battery units.

\bibliographystyle{IEEEtran}
\bibliography{Ref}

\end{document}